\documentclass{article}

\usepackage[cmex10]{amsmath}
\usepackage{amssymb}
\usepackage{amsthm}
\usepackage{enumerate}
\usepackage{amscd}
\usepackage{url}
\usepackage{mathrsfs}

\theoremstyle{plain}
\newtheorem{theorem}{Theorem}
\newtheorem{lemma}[theorem]{Lemma}
\newtheorem{proposition}[theorem]{Proposition}
\newtheorem{corollary}[theorem]{Corollary}
\newtheorem{definition}[theorem]{Definition}
\newtheorem{problem}[theorem]{Problem}

\theoremstyle{definition}

\newcommand{\beq}{\begin{equation*}}
\newcommand{\eeq}{\end{equation*}}

\DeclareMathOperator{\tr}{Tr}
\DeclareMathOperator{\rk}{rk}
\DeclareMathOperator{\im}{im}
\DeclareMathOperator{\codim}{codim}
\DeclareMathOperator{\dmin}{d_{min}}
\DeclareMathOperator{\dmax}{d_{max}}
\DeclareMathOperator{\Bilin}{Bilin}

\newcommand{\F}{\mathbb{F}}
\newcommand{\cL}{\mathscr{L}}
\newcommand{\cP}{\mathcal{P}}
\newcommand{\PP}{\mathbb{P}}
\newcommand{\EE}{\mathbb{E}}
\newcommand{\XX}{X_{k\times\ell}}

\newcommand{\eg}{\emph{e.g. }}
\newcommand{\longto}{\longrightarrow}

\newcommand{\tens}{\otimes}
\newcommand{\moins}{\setminus}

\newcommand{\ev}{\mathrm{ev}}
\newcommand{\deux}[1][2]{^{\langle #1\rangle}}
\newcommand{\abs}[1]{\lvert #1\rvert}
\newcommand{\linspan}[1]{\langle #1\rangle}
\newcommand{\SFrob}[1]{S^{#1}_{\textrm{Frob}}}
\newcommand{\gauss}[2]{\begin{bmatrix}#1\\#2\end{bmatrix}_q}

\let\phi\varphi
\let\epsilon\varepsilon
\let\subset\subseteq
\let\ell l

\begin{document}

\title{Linear independence of rank~$1$ matrices\\ and the dimension of $*$-products of codes}

\author{Hugues Randriambololona\\
ENST ``Telecom ParisTech'' \& LTCI CNRS UMR 5141}

\maketitle

\begin{abstract}
We show that with high probability, random rank~$1$ matrices over a finite field are in (linearly) general position,
at least provided their shape $k\times\ell$ is not excessively unbalanced.
This translates into saying that the dimension of the $*$-product of two $[n,k]$ and $[n,\ell]$ random codes is equal
to $\min(n,kl)$, as one would have expected.
Our work is inspired by a similar result of Cascudo-Cramer-Mirandola-Z\'emor \cite{CCMZ} dealing with $*$-squares of codes,
which it complements, especially regarding applications to the analysis of McEliece-type cryptosystems \cite{CGGOT}\cite{COT}.
We also briefly mention the case of higher $*$-powers, which require to take the Frobenius into account.
We then conclude with some open problems.
\end{abstract}


\section{Introduction}
\label{Intro}
Many fundamental problems in information theory and in theoretical computer science can be expressed in terms of the
structure of linearly independent and generating subsets of a set in a vector space, as illustrated by \cite{W} and the subsequent success 
of matroid theory.
In this context the importance of the following definition is self-evident:
\setcounter{theorem}{-1}
\begin{definition}
\label{gp}
Let $V$ be a finite-dimensional vector space, over an arbitrary field.
We say a set $X\subset V$ is in general position if
any finite subset $S\subset X$ has its linear span $\linspan{S}$ of dimension
\beq
\dim\linspan{S}=\min(\abs{S},\dim V).
\eeq
\end{definition}
This means that there are no more linear relations than expected between elements of $X$:
any $S\subset X$ of size $\abs{S}\leq\dim V$ is linearly independent,
and any $S\subset X$ of size $\abs{S}\geq\dim V$ is a generating set in $V$.

This requirement is quite strong, and weaker variants have been considered. We can cite at least three of them.

The first one is to introduce thresholds. We say $X$ is in $(a,b)$-general position
if any $S\subset X$ of size $\abs{S}\leq a$ is linearly independent,
and any $S\subset X$ of size $\abs{S}\geq b$ is a generating set in $V$.
This notion should look very familiar to coding experts. Indeed one shows easily:
\begin{lemma}
\label{gp-dmin}
Let $C$ be a $q$-ary $[n,k]$ code, with generating matrix $G$.
Set $V=\F_q^k$ and let $X\subset V$ be the set of columns of $G$.
Then $X$ is in $(a,b)$-general position, with $a=\dmin(C^\perp)-1$
and $b=n-\dmin(C)+1$.
\end{lemma}

A second one is to allow a small gap $g$ from the expected dimension: we say $X$ is in $g$-almost general position
if for any $S\subset X$ we have
\beq
\dim\linspan{S}\geq\min(\abs{S},\dim V)-g.
\eeq
This means allowing up to $g$ more linear relations than expected.
There is an obvious link with the previous notion:
\begin{lemma}
\label{gp-gap}
If $X\subset V$ is in $(a,b)$-general position, then it is in $g$-almost general position
for $g=\min(\dim V-a,b-\dim V)$.
\end{lemma}
We leave it to the reader to combine Lemma~\ref{gp-dmin} and Lemma~\ref{gp-gap} and give a coding-theoretic interpretation of this integer $g$ (or a geometric interpretation in case $C$ is an AG-code).

Last, our third variant is probabilistic, allowing a small proportion of $S$ to fail in Definition~\ref{gp}.
In fact, rather than subsets of $X$, it will be easier to consider sequences of elements of $X$, possibly with repetitions.
For this we will assume that $X$ is equipped with a probability distribution $\cL$.
A natural choice when $X$ is finite would be to take the uniform distribution,
however more general $\cL$ will be allowed.
Then, measuring how close $X\subset V$ is to being in general position reduces to the following:
\begin{problem}
\label{gp-proba}
Let $n\geq1$, and $\mathbf{u_1},\dots,\mathbf{u_n}$ random elements of~$X$
(understood: independent, and distributed according to $\cL$).
Give bounds on the ``error probability''
\beq
\PP[\dim\linspan{\mathbf{u_1},\dots,\mathbf{u_n}}<\min(n,\dim V)].
\eeq
\end{problem}

In this work we address this problem for $V=\F_q^{k\times\ell}$ a matrix space,
and $X\subset V$ the set of matrices of rank~$1$.

Understanding the linear span of families of rank~$1$ matrices is especially important
regarding the theory of bilinear complexity (or equivalently, that of tensor decomposition).
Indeed, computing the complexity of a bilinear map (or the rank of a $3$-tensor) reduces to the following \cite{BDEZ}\cite{BD}\cite{LW}\cite{ChCh+}:
given a linear subspace $W\subset\F_q^{k\times\ell}$, find a family of rank~$1$ matrices of minimal cardinality
whose linear span contains $W$.

Another motivation comes from the theory of $*$-products of codes,
and in particular its use in a certain class of attacks \cite{CGGOT}\cite{COT} against McEliece-type cryptosystems.
Given words $\mathbf{c}=(c_1,\dots,c_n),\mathbf{c'}=(c'_1,\dots,c'_n)\in\F_q^n$,
we let $\mathbf{c}*\mathbf{c'}=(c_1c'_1,\dots,c_nc'_n)\in\F_q^n$ be their componentwise product.
Then \cite{AGCT} if $C,C'\subset\F_q^n$ are two linear codes of the same length,
their product $C*C'\subset\F_q^n$ is defined as the linear span of the $\mathbf{c}*\mathbf{c'}$ for $\mathbf{c}\in C,\mathbf{c'}\in C'$.
We can also define the square $C\deux=C*C$, and likewise for higher powers $C\deux[j]$.

Setting $k=\dim C$ and $\ell=\dim C'$, it is then easily seen
\beq
\dim C*C'\leq k\ell,
\eeq
\beq
\dim C\deux\leq k(k+1)/2,
\eeq
and in fact for small $k,\ell$ and random $C,C'$ one expects these inequalities to be equalities.
For the second inequality, this is proved in \cite{CCMZ}. For the first inequality, we will see this reduces
to our solution of Problem~\ref{gp-proba} for rank~$1$ matrices.

So, together, \cite{CCMZ} and our results support the heuristic at the heart of the aforementionned attacks against McEliece-type cryptosystems.
Indeed, the very principle of these attacks is to uncover the hidden algebraic structure of an apparently random code (which serves as the public key)
by identifying subcodes for which equality fails in these inequalities
(for instance, the dimension of the product behaving additively rather than multiplicatively).

\section{Generic approach}

Here $V$ is an abstract vector space of dimension $m$ over $\F_q$, and $X\subset V$ an arbitrary subset. We may assume $X$ spans $V$.
We are interested in the function
\beq
\PP(n)=\PP[\dim\linspan{\mathbf{u_1},\dots,\mathbf{u_n}}<\min(n,\dim V)].
\eeq
Clearly it is unimodal, more precisely it is increasing for $n\leq m$ and decreasing for $n\geq m$.
Now we study each of these two cases in more detail.

\subsection{Case $n\geq m$.}

We have $\dim\linspan{\mathbf{u_1},\dots,\mathbf{u_n}}<m$ iff $\mathbf{u_1},\dots,\mathbf{u_n}$ are contained in an hyperplane $H$ of $V$.
Using the union bound and the independence of the $\mathbf{u_i}$ we get at once:
\begin{proposition}
\label{union}
We have
\beq
\begin{split}
\PP(n)&\leq\sum_H\PP[\mathbf{u_1},\dots,\mathbf{u_n}\in H]\\
&=\sum_H\PP[\mathbf{u_1}\in H]^n
\end{split}
\eeq
where $H$ ranges over hyperplanes of $V$.
\end{proposition}

This bound is exponentially small. More precisely, set
\beq
\rho=\max_H\PP[\mathbf{u_1}\in H]
\eeq
(for instance $\rho=\max_H\abs{X\cap H}/\abs{X}$ if $\cL$ is uniform distribution).
We then see immediately:

\begin{corollary}
\label{cc'}
For all $n\geq m$ we have
\beq
c\rho^{n-m}\leq\PP(n)\leq c'\rho^{n-m}.
\eeq
where $c=\rho^m$ and $c'=\sum_H\PP[\mathbf{u_1}\in H]^m$.
\end{corollary}

It should be noted that $c,c',\rho$ depend on $V$ and $X$.
So, part of the job will be to make these constants more explicit when $V$ and $X$ will be specified.

Another interesting fact is that the RHS in Proposition~\ref{union} is
\beq
\sum_H\PP[\mathbf{u_1},\dots,\mathbf{u_n}\in H]=\EE[\abs{\{H;\;\mathbf{u_1},\dots,\mathbf{u_n}\in H\}}],
\eeq
the expected value of the number of hyperplanes containing $\mathbf{u_1},\dots,\mathbf{u_n}$. However, this number is precisely $\frac{q^d-1}{q-1}$,
where $d=\codim\linspan{\mathbf{u_1},\dots,\mathbf{u_n}}$.
This allows us to combine our second and third variants of the notion of general position:
\begin{proposition}
For $0\leq g\leq\min(m,n)$ we have
\beq
\PP[\dim\linspan{\mathbf{u_1},\dots,\mathbf{u_n}}<m-g]\leq c'\rho^{n-m}\frac{q-1}{q^{g+1}-1}
\eeq
(with $c',\rho$ as above), and also
\beq
\PP[\dim\linspan{\mathbf{u_1},\dots,\mathbf{u_n}}<m-g]\leq\sum_W\PP[\mathbf{u_1}\in W]^n
\eeq
where $W\subset V$ ranges over subspaces of codimension $g+1$.
\end{proposition}
\begin{proof}
The first inequality follows from the discussion above, using Markov's inequality as in \cite[Prop.~5.1]{CCMZ}.
The second is a direct approach using the union bound similar to that of Proposition~\ref{union}.
\end{proof}
Which of these two bounds is stronger, and which is more tractable, certainly depends on $V$ and $X$. 
Note also that the bounds remain valid even without the assumption $m\leq n$.

We illustrate what precedes for $X=V=\F_q^m$ with uniform distribution (this will be used later).
We introduce the converging infinite product
\beq
C_q=\prod_{j\geq1}(1-q^{-j})^{-1}.
\eeq
Numerically, $C_q\leq C_2\approx 3.463.$

We let $\gauss{m}{r}=\prod_{1\leq j\leq r}\frac{q^{m-r+j}-1}{q^j-1}$ denote the number of $r$-dimensional subspaces in $\F_q^m$.
\begin{lemma}
\label{bgc}
We have
\beq
q^{r(m-r)}\leq\gauss{m}{r}\leq C_q q^{r(m-r)}.
\eeq
\end{lemma}
\begin{proof}
From $q^{m-r}\leq\frac{q^{m-r+j}-1}{q^j-1}\leq(1-q^{-j})^{-1}q^{m-r}$.
\end{proof}

\begin{proposition}
\label{toy}
For $0\leq r\leq\min(m,n)$ and random $\mathbf{u_1},\dots,\mathbf{u_n}\in\F_q^m$ uniformly distributed, we have
\beq
\PP[\dim\linspan{\mathbf{u_1},\dots,\mathbf{u_n}}\leq r]\leq C_q q^{-(n-r)(m-r)}.
\eeq
\end{proposition}
\begin{proof}
Follows from what precedes, using $\PP[\mathbf{u_1}\in W]=q^{-(m-r)}$ for $\dim W=r$.
\end{proof}

\subsection{Case $n\leq m$.}

From now on we will suppose $(X,\cL)$ is homothety invariant: given any $\lambda\in\F_q^\times$,
then for random $\mathbf{u}\in X$, we also have $\lambda\mathbf{u}\in X$, with the same distribution $\cL$.

We say a vector $\mathbf{z}=(\lambda_1,\dots,\lambda_n)\in\F_q^n$ is a linear relation for $\mathbf{u_1},\dots,\mathbf{u_n}$
if $\lambda_1\mathbf{u_1}+\cdots+\lambda_n\mathbf{u_n}=0$.

Also introduce the random variable
\beq
\mathbf{s_n}=\mathbf{u_1}+\cdots+\mathbf{u_n}\in V.
\eeq

\begin{lemma}
For any $\mathbf{z}\in\F_q^n$ of Hamming weight $w$, we have
\beq
\PP[\textrm{$\mathbf{z}$ is a linear relation for $\mathbf{u_1},\dots,\mathbf{u_n}$}]\;=\;\PP[\mathbf{s_w}=0].
\eeq
\end{lemma}
\begin{proof}
We may suppose $\mathbf{z}$ has support $\{1,\dots,w\}$, and we conclude since $\mathbf{u_i}$ and $\lambda_i\mathbf{u_i}$
have same distribution for $\lambda_i\neq0$.
\end{proof}

\begin{proposition}
\label{ssw}
We have
\beq
\PP(n)\leq\sum_{w\geq1}\binom{n}{w}(q-1)^{w-1}\PP[\mathbf{s_w}=0]
\eeq
\end{proposition}
\begin{proof}
Union bound, as in Proposition~\ref{union} (note that we may count linear relations only up to proportionality).
\end{proof}

Likewise, Markov's inequality gives, for any $g\geq0$,
\beq
\PP(\dim\linspan{\mathbf{u_1},\!\dots,\!\mathbf{u_n}}\!<\!n-g)\leq{\textstyle\frac{1}{q^{g+1}\!-\!1}}\sum_{w\geq1}\!{\textstyle\binom{n}{w}}(q-1)^w\PP[\mathbf{s_w}\!=\!0].
\eeq

In these sums we expect the contribution of linear relations of large weight should stay under control thanks to:
\begin{proposition}
\label{asympt}
As $w\to\infty$ we have
\beq
\PP[\mathbf{s_w}=0]\to\frac{1}{q^m},
\eeq
except for $q=2$ and $X$ contained in the translate of an hyperplane, in which case we have $\PP[\mathbf{s_w}=0]$ for odd $w$,
and $\PP[\mathbf{s_w}=0]\to\frac{1}{2^{m-1}}$ for even $w\to\infty$.
\end{proposition}
\begin{proof}
We treat first the case $q>2$, so there is a $\lambda\neq0,1$ in $\F_q$.
The $\mathbf{s_w}$ form a random walk on the finite commutative group $V$.
Seen as a Markov chain, it is irreducible, because $X$ spans $V$ (as a vector space, but also as a group, since $X$ is homothety-invariant).
Moreover it is aperiodic, because the zero vector can be written as a sum of $2$ elements of $X$ (\eg $\mathbf{s}+(-\mathbf{s})$),
and also as a sum of $3$ elements of $X$ (\eg $(1-\lambda)\mathbf{s}+(-\mathbf{s})+\lambda\mathbf{s}$).
So it converges to its unique stationnary distribution, which can only be uniform.

%
The case $q=2$ is similar, with a tweak on aperiodicity.
\end{proof}

\section{Rank~$1$ matrices}

A matrix $\mathbf{u}\in\F_q^{k\times\ell}$ is of rank~$1$ iff it can be written $\mathbf{u}=\mathbf{p}\mathbf{q}^T$
for column vectors $\mathbf{p}\in\F_q^k\moins\{\mathbf{0}\}$, $\mathbf{q}\in\F_q^\ell\moins\{\mathbf{0}\}$.
Moreover these $\mathbf{p},\mathbf{q}$ are uniquely determined up to a scalar.
This means, choosing random $\mathbf{p}\in\F_q^k\moins\{\mathbf{0}\}$, $\mathbf{q}\in\F_q^\ell\moins\{\mathbf{0}\}$ uniformly,
and setting $\mathbf{u}=\mathbf{p}\mathbf{q}^T$, gives a random matrix of rank~$1$ with uniform distribution.

Actually we will use a slightly different model. Let 
\beq
\XX=\{\mathbf{u}\in\F_q^{k\times\ell};\;\rk\mathbf{u}\leq 1\}
\eeq
be the set of rank~$1$ matrices together with the zero matrix.
Pick random $\mathbf{p}\in\F_q^k$, $\mathbf{q}\in\F_q^\ell$ uniformly (possibly zero), 
and set $\mathbf{u}=\mathbf{p}\mathbf{q}^T$. This gives our distribution $\cL$ on $\XX$.

Note that if $\mathbf{u}\in \XX$ is distributed according to $\cL$,
then conditioning on the event $\mathbf{u}\neq\mathbf{0}$ gives back the uniform distribution on matrices of rank~$1$.
Conversely, if $b$ is a Bernoulli variable of parameter $\PP[b=1]=(1-q^{-k})(1-q^{-\ell})$, and if $\mathbf{u}$ is a random
uniformly distributed matrix of rank~$1$, then $b\mathbf{u}\in \XX$ is distributed according to $\cL$.
Moreover, replacing $\mathbf{u_1},\dots,\mathbf{u_n}$ with $b_1\mathbf{u_1},\dots,b_n\mathbf{u_n}$ can only decrease the dimension
of their linear span.
As a consequence, any upper bound on $\PP(n)$ for $(\XX,\cL)$ will also be an upper bound for uniformly distributed matrices of rank~$1$.


\begin{lemma}
\label{bilin}
\begin{enumerate}[(i)]
\item
Every linear form on $\F_q^{k\times\ell}$ is of the form $l_{\mathbf{B}}=\tr(\mathbf{B}^T\cdot)$ for a uniquely determined $\mathbf{B}\in\F_q^{k\times\ell}$. 
\item
The number of $\mathbf{B}\in\F_q^{k\times\ell}$ of rank $r$ is
\beq\textstyle
\gauss{k}{r}\gauss{\ell}{r}\abs{\mathrm{GL}_r(\F_q)}\leq C_qq^{r(k+\ell-r)}.
\eeq
\item
Given $\mathbf{B}\in\F_q^{k\times\ell}$ of rank $r$, then for random $\mathbf{u}=\mathbf{p}\mathbf{q}^T$ in $\XX$ we have
$\PP[l_{\mathbf{B}}(\mathbf{u})=0]=\frac{1}{q}\left(1+\frac{q-1}{q^r}\right)$.
\end{enumerate}
\end{lemma}
\begin{proof}
Point (i) is clear. For point (ii) we view $\mathbf{B}$ as a linear map $\F_q^k\longto\F_q^\ell$,
and we note that it is entirely determined by
its kernel $\ker\mathbf{B}\subset\F_q^k$ of codimension $r$, 
its image $\im\mathbf{B}\subset\F_q^\ell$ of dimension $r$, 
and the isomorphism $\F_q^k/\ker\mathbf{B}\simeq\im\mathbf{B}$ it induces. 
This gives the formula of the LHS, and the upper bound works as in the proof of Lemma~\ref{bgc}.
For (iii) we note $l_{\mathbf{B}}(\mathbf{u})=0$ means $\mathbf{p}^T\mathbf{B}\mathbf{q}=0$, which happens
precisely when $\mathbf{p}^T\mathbf{B}=0$ (of probability $q^{-r}$) or when $\mathbf{q}$ is orthogonal to $\mathbf{p}^T\mathbf{B}\neq0$
(of probability $q^{-1}(1-q^{-r})$).
\end{proof}

For some of our results we will restrict to matrices whose long side grows at most exponentially in the short side.
More precisely, for any $\epsilon,\kappa>0$, we introduce the parameter space
\beq
\cP(\epsilon,\kappa)=\left\{(k,\ell);\;2\leq k\leq\ell\leq\frac{\epsilon q^{\kappa k}}{(q-1)k}\right\}.
\eeq

Now we fix a $\kappa>0$ small enough so that $q^{(1-\kappa)^2}\geq1+\frac{q-1}{q}$ (for instance $\kappa=0.23$ works for any $q$),
as well as some $0<\epsilon<1$.
\begin{theorem}
\label{n>kl}
Let $(k,\ell)\in\cP(\epsilon,\kappa)$ and $n\geq k\ell$. Then for random $\mathbf{u_1},\dots,\mathbf{u_n}\in \XX$ we have
\beq
\PP[\mathbf{u_1},\dots,\mathbf{u_n}\textrm{ don't span }\F_q^{k\times\ell}]\leq c''\rho^{n-k\ell}
\eeq
with $\rho=\frac{1}{q}\left(1+\frac{q-1}{q}\right)$ and $c''=\frac{qC_q}{(q-1)^2}\left(1+\frac{1}{1-\epsilon}\right)$.
\end{theorem}
\begin{proof}
We apply Corollary~\ref{cc'}, where from Lemma~\ref{bilin} we get $\rho=\frac{1}{q}\left(1+\frac{q-1}{q}\right)$ and
\beq
\begin{split}
c'&\leq\frac{1}{q-1}\sum_{1\leq r\leq k}C_qq^{r(k+\ell-r)}\textstyle\left(\frac{1}{q}\left(1+\frac{q-1}{q^r}\right)\right)^{k\ell}\\
&=\frac{C_q}{q-1}\sum_{1\leq r\leq k}\textstyle\frac{\left(1+\frac{q-1}{q^r}\right)^{k\ell}}{q^{(k-r)(\ell-r)}}.
\end{split}
\eeq
We set $r_0=\lfloor\kappa k\rfloor$ and split this last sum in two.

First, for $r\leq r_0$ we have $(k-r)(\ell-r)\geq(1-\kappa)^2k\ell+(r_0-r)$ and $1+\frac{q-1}{q^r}\leq1+\frac{q-1}q$,
so, by our condition on $\kappa$, $\frac{\left(1+\frac{q-1}{q^r}\right)^{k\ell}}{q^{(k-r)(\ell-r)}}\leq\frac{1}{q^{r_0-r}}$.

On the other hand, for $r>r_0$ we have $\left(1+\frac{q-1}{q^r}\right)^{k\ell}<\left(1+\frac{q-1}{q^{\kappa k}}\right)^{k\ell}\leq\frac{1}{1-\frac{k\ell(q-1)}{q^{\kappa k}}}\leq\frac{1}{1-\epsilon}$.

We deduce:
\beq
c'<\frac{C_q}{q-1}\left(\sum_{1\leq r\leq r_0}\!\frac{1}{q^{r_0-r}}+\frac{1}{1\!-\!\epsilon}\sum_{r_0< r\leq k}\!\frac{1}{q^{(k-r)(\ell-r)}}\right)\leq c''.
\eeq
\end{proof}

Given $k\leq\ell$ and random $\mathbf{u_i}\in\XX$, recall for all $w\geq1$ we set
$\mathbf{s_w}=\mathbf{u_1}+\cdots+\mathbf{u_w}\in\F_q^{k\times\ell}$.
\begin{theorem}
\label{Psw}
\begin{enumerate}[(i)]
\item For $1\leq w<k+\ell$ we have
\beq
\PP[\mathbf{s_w}=0]\leq\frac{2qC_q/(q-1)}{q^{kw/2}}.
\eeq
\item For $w\geq k+\ell$ we have
\beq
\PP[\mathbf{s_w}=0]\leq\frac{C_q(1-q^{-(w-\ell)})^{-1}}{q^{k\ell}}.
\eeq
\end{enumerate}
\end{theorem}
\begin{proof}
Write $\mathbf{u_i}=\mathbf{p}_i\mathbf{q_i}^T$ with $\mathbf{p_i}\in\F_q^k$, $\mathbf{q_i}\in\F_q^\ell$ uniform.
Let $G$ be the $k\times w$ matrix whose columns are $\mathbf{p_1},\dots,\mathbf{p_w}$, and let $\mathbf{x_1},\dots,\mathbf{x_k}\in\F_q^w$ be its rows.
Likewise let $G'$ be the $\ell\times w$ matrix whose columns are $\mathbf{q_1},\dots,\mathbf{q_w}$, and let $\mathbf{y_1},\dots,\mathbf{y_\ell}\in\F_q^w$ be its rows.
Note these $\mathbf{x}$'s and $\mathbf{y}$'s are uniform and independent.
Also our key observation is that $\mathbf{s_w}=0$ iff
$\linspan{\mathbf{x_1},\dots,\mathbf{x_k}}\perp\linspan{\mathbf{y_1},\dots,\mathbf{y_\ell}}$ in $\F_q^w$.

Now we condition on $\dim\linspan{\mathbf{y_1},\dots,\mathbf{y_\ell}}$.

By Proposition~\ref{toy} we have $\PP[\dim\linspan{\mathbf{y_1},\dots,\mathbf{y_\ell}}=e]\leq C_qq^{-(\ell-e)(w-e)}$.
Also, $\PP[\linspan{\mathbf{x_1},\dots,\mathbf{x_k}}\!\!\perp\!\!\linspan{\mathbf{y_1},\dots,\mathbf{y_\ell}}|\dim\linspan{\mathbf{y_1},\dots,\mathbf{y_\ell}}\!=\!e]=q^{-ke}$.
This gives
\beq
\PP[\mathbf{s_w}=0]\leq C_q\sum_{0\leq e\leq\min(\ell,w)}q^{-f(e)}
\eeq
where $f(e)=ke+(\ell-e)(w-e)$. This function $f$ attains its minimum at $e_0=(\ell+w-k)/2$, from which we deduce, for $0\leq e\leq\min(\ell,w)$:
\beq
f(e)\geq
\begin{cases}
kw+(w-e)\geq kw/2+(w-e) & \textrm{for $w\leq\ell-k$}\\
f(e_0)+\lfloor\abs{e-e_0}\rfloor\geq kw/2+\lfloor\abs{e-e_0}\rfloor  & \textrm{for $\ell-k<w<k+\ell$}\\
kl+(e-k)(w-\ell) & \textrm{for $w\geq k+\ell$.}
\end{cases}
\eeq
The first two cases together give point (i), while the third gives point (ii).
\end{proof}

\begin{theorem}
\label{n<kl}
Let $(k,\ell)\in\cP(\epsilon,\frac{1}{2})$ and $n\leq k\ell$. Then for random $\mathbf{u_1},\dots,\mathbf{u_n}\in \XX$ we have
\beq
\PP[\mathbf{u_1},\dots,\mathbf{u_n}\textrm{ lin. dependent}]\leq \frac{qC_q}{(q-1)^2}\!\left(\frac{2\epsilon}{1-\epsilon}+q^{-(k\ell-n)}\!\right)\!\!.
\eeq
\end{theorem}
\begin{proof}
Split the sum in Proposition~\ref{ssw} in two: for $w<k+\ell$ use Theorem~\ref{Psw}(i) and $\binom{n}{w}\leq(k\ell)^w$;
for $w\geq k+\ell$ use Theorem~\ref{Psw}(ii).
\end{proof}

\section{Products of codes}

By a generating matrix for a linear code $C$ we mean any matrix $G$ whose row span is $C$. We allow $G$ to have more than $\dim C$ rows.

Consider random $\mathbf{G}\in\F_q^{k\times n}$, $\mathbf{G'}\in\F_q^{\ell\times n}$ (uniform distribution),
generating matrices for $C,C'\subset\F_q^n$, so $\dim C\leq k$, $\dim C'\leq\ell$.
Denote by $\mathbf{p_1},\dots,\mathbf{p_n}\in\F_q^k$ the columns and by $\mathbf{x_1},\dots,\mathbf{x_k}\in\F_q^n$ the rows of $G$.
Denote by $\mathbf{q_1},\dots,\mathbf{q_n}\in\F_q^\ell$ the columns and by $\mathbf{y_1},\dots,\mathbf{y_\ell}\in\F_q^n$ the rows of $G'$.

Identify the matrix space $\F_q^{k\times\ell}$ with $\F_q^{k\ell}$ .

The product $C*C'$ and its generating matrix $\widehat{G}\in\F_q^{(k\times\ell)\times n}$ admit the following equivalent descriptions \cite{AGCT}:
\begin{enumerate}[(i)]
\item $\widehat{G}$ has rows all products $\mathbf{x_i}*\mathbf{y_j}$
\item $C*C'$ is the projection of $C\tens C'$ on the diagonal
\item $\widehat{G}$ has columns 
the $\rk\leq\!1$ matrices $\mathbf{p_1}\mathbf{q_1}^T,\dots,\mathbf{p_n}\mathbf{q_n}^T$
\item $C*C'$ is the image of the evaluation map\\
$\begin{matrix}\ev\!:\! & \!\!\Bilin(\F_q^k\times\F_q^\ell)\!\! & \longto & \F_q^n\\ & B & \mapsto & \!\!\!(B(\mathbf{p_1},\!\mathbf{q_1}),\dots,B(\mathbf{p_n},\!\mathbf{q_n})).\end{matrix}$
\end{enumerate}

From description (iii) we can translate our Theorems~\ref{n>kl} and~\ref{n<kl}.
Recall $q^{(1-\kappa)^2}\geq1+\frac{q-1}{q}$, and $0<\epsilon<1$.
\begin{theorem}
\label{n>kl_alt}
For $(k,\ell)\in\cP(\epsilon,\kappa)$ and $n\geq k\ell$, we have
\beq
\PP[\dim C*C'<kl]\leq c''\rho^{n-k\ell}
\eeq
with $\rho=\frac{1}{q}\left(1+\frac{q-1}{q}\right)$ and $c''=\frac{qC_q}{(q-1)^2}\left(1+\frac{1}{1-\epsilon}\right)$.
\end{theorem}
\begin{theorem}
\label{n<kl_alt}
For $(k,\ell)\in\cP(\epsilon,\frac{1}{2})$ and $n\leq k\ell$, we have
\beq
\PP[\dim C*C'<n]\leq \frac{qC_q}{(q-1)^2}\!\left(\frac{2\epsilon}{1-\epsilon}+q^{-(k\ell-n)}\!\right)\!\!.
\eeq
\end{theorem}

Note that if $k\to\infty$ and $k\ell/q^{k/2}\to0$ (for instance if $\ell$ is polynomial in $k$), we can set $\epsilon=(q-1)k\ell/q^{k/2}\to0$.

Still, we can derive an unconditional result, valid for any $(k,\ell)$. Recall the maximum distance $\dmax$ of a linear code
is the \emph{maximum} weight of a codeword.
\begin{theorem}
\label{dmaxdual}
For any $(k,\ell)$, and  $k+\ell\leq n\leq k\ell$, we have
\beq
\PP[\dmax(C*C')^\perp\geq k+\ell]\leq\frac{qC_q}{(q-1)^2}q^{-(k\ell-n)}.
\eeq
\end{theorem}
\begin{proof}
Union bound for $\PP[\exists\textrm{lin. rel. of weight }\geq k+\ell]$,
which means keep only terms $w\geq k+\ell$ in Proposition~\ref{ssw}, and use only part (ii) of Theorem~\ref{Psw}.
\end{proof}
So, with high probability $(C*C')^\perp$ has $\dmax\!<\!k+\ell$. This is a strong restriction (for instance it also implies $\dim\!<\!k+\ell$).

\section{Squares and higher powers}

Let $C$ have generating matrix $\mathbf{G}\in\F_q^{k\times n}$, with columns $\mathbf{p_1},\dots,\mathbf{p_n}\in\F_q^k$ and rows $\mathbf{x_1},\dots,\mathbf{x_k}\in\F_q^n$.
As above, the $s$-th power $C\deux[s]$ and its generating matrix $\widehat{G}$ admit the following equivalent descriptions:
\begin{enumerate}[(i)]
\item $\widehat{G}$ has rows all $*$-monomials of degree $s$ in the $\mathbf{x_i}$
\item $C\deux[s]$ is the projection of $C^{\tens s}$ on the diagonal
\item $\widehat{G}$ has columns 
the elementary tensors $\mathbf{p_1}^{\tens s},\dots,\mathbf{p_n}^{\tens s}$
\item $C\deux[s]$ is the image of the evaluation map\\
$\begin{matrix}\ev\!:\! & \!\!\F_q[t_1,\dots,t_k]_s & \longto & \F_q^n\\ & P & \mapsto & \!\!\!(P(\mathbf{p_1}),\dots,P(\mathbf{p_n})).\end{matrix}$
\end{enumerate}
(Where $R_s$ denotes the $s$-th homogeneous component of $R$.)

We deduce at once $\dim C\deux[s]\leq\min(n,\binom{k+s-1}{s})$.
For $s=2$ it is shown in \cite{CCMZ} that for random such $C$, with high probability there is equality: $\dim C\deux=\min(n,\frac{k(k+1)}{2})$
(which could in turn be translated into a general position result for rank~$1$ symmetric matrices).
It is interesting to note that not having to face unbalanced $(k,\ell)$ made it easier for these authors to deal with short relations, hence to control $\dmin(C\deux)^\perp$ in \cite[Prop.~2.4]{CCMZ}.
By contrast, in our setting, independence of $C$ and $C'$ made it easier to deal with long relations, hence to control $\dmax(C*C')^\perp$ in Theorem~\ref{dmaxdual}.

Concerning higher powers, one should be careful of the:
\begin{proposition}
\label{strict}
For $s>q$ we always have strict inequality 
\beq
\dim C\deux[s]<\binom{k+s-1}{s}.
\eeq
More precisely, we have
\beq\dim C\deux[s]\leq\min(n,\chi_q(k,s))\eeq
where \cite[App.~A]{AGCT}:
\beq\textstyle\chi_q(k,s)=\dim\SFrob{s}\F_q^k=\dim(\F_q[t_1,\dots,t_k]/(t_i^qt_j-t_it_j^q))_s.\eeq
\end{proposition}
\begin{proof}
The map $*$ is Frobenius-symmetric, so in (ii) the projection $C^{\tens s}\to C\deux[s]$ factors through $\SFrob{s}C$.
Alternatively, in (iv), $\ker(\ev)$ contains all multiples of the $t_i^qt_j-t_it_j^q$.
\end{proof}

\section{Open problems}

In our probabilistic model we considered random matrices of the form $\mathbf{u_i}=\mathbf{p_i}\mathbf{q_i}^T$
for column vectors $\mathbf{p_i}\in\F_q^k$, $\mathbf{q_i}\in\F_q^\ell$ possibly zero.
However, as already noted, it is perhaps more natural to restrict these $\mathbf{p_i},\mathbf{q_i}$ to stay nonzero,
so the $\mathbf{u_i}$ become uniformly distributed rank~$1$ matrices.
Considering the $\mathbf{p_i}$ (resp. $\mathbf{q_i}$) as the columns of a generating matrix of a code $C$
(resp. $C'$), this translates into considering only codes with full support---although of dimension possibly less than $k$ (resp. $\ell$).
Then, a further model would be to request these generating matrices having full rank. 
That means: take $C$ (resp. $C'$) uniformly distributed in the set of $[n,k]$ (resp. $[n,\ell]$) codes with full support.
Clearly this could only help get sharper bounds. In particular:

\begin{problem}
Do these alternative models allow to relax our condition $\cP(\epsilon,\kappa)$? Do they give bounds valid without any restriction on $(k,\ell)$?
\end{problem}

Proposition~\ref{asympt} suggests that the fate of long relations should essentially not depend on the probabilistic model.
On the other hand, for short relations, it certainly does.
In fact, relations of weight less than $k+\ell$ are perhaps less tractable because, for such a length, $C$ and $C'$ necessarily intersect.
This leads to the following, which would encompass both our results (remove the conditioning) and those of \cite{CCMZ} (set $i=k=\ell$):
\begin{problem}
For any $n,k,\ell,i,j$, estimate the conditional probability
\beq
\PP[\dim C*C'=j|\dim C\cap C'=i].
\eeq
\end{problem}

We saw the existence of relations of length $w$ is related to the distribution of $\mathbf{s_w}=\mathbf{u_1}+\cdots+\mathbf{u_w}$.
When the $\mathbf{u_i}$ are uniformly distributed matrices of rank~$1$, this reduces to:
\begin{problem}
In $\F_q^{k\times\ell}$, what is the number
\beq
N^{k\times\ell}_q(r,w)
\eeq
of decompositions of a matrix of rank~$r$ as an ordered sum of $w$ matrices of rank~$1$?
\end{problem}
It is easily seen that this number is well defined, which means, it is the same for all such matrices of rank~$r$. 
Of special importance are the $N^{k\times\ell}_q(0,w)$, which control $\PP[\mathbf{s_w}=0]$.
We leave it as an exercise to link their computation with that of the weight distribution of the code
$(S_k\tens S_\ell)^\perp$,
where $S_k$ is the $[\frac{q^k-1}{q-1},k]$ $q$-ary simplex code.

Considering powers of a code leads similarly to count families of elementary $s$-th power tensors summing to zero.
\begin{problem}
For fixed $s$, and a random $[n,k]$ code $C$, estimate the probability $\PP[\dim C\deux[s]=\min(n,\chi_q(k,s))]$.

And then, what if we also let $s$ vary?
\end{problem}
It is interesting to note that, up to code equivalence, any $[n,k]$ code $C$ with full support can be obtained
from the simplex code $S_k$ by deleting and repeating columns. Then $C\deux[s]$ is obtained from $S_k\deux[s]$
by deleting and repeating the same columns. Some authors also call $S_k\deux[s]$ the $s$-th order projective Reed-Muller
code (in $k$ variables); it has dimension $\chi_q(k,s)$.
As above, we can split our Problem in two cases: for $n\geq\chi_q(k,s)$, we're interested in relations between rows of the
generating matrix of $C$, which is linked to the weight distribution of $S_k\deux[s]$;
while for $n\leq\chi_q(k,s)$, we're interested in relations between columns, which is linked to its dual weight distribution.

Last, it is the author's opinion that considering only the dimension of products is not entirely in the spirit of coding theory.
In fact, it is a purely algebraic problem, where $(\F_q^n,*)$ could be replaced by any space equipped with a bilinear inner composition law.
See \cite{BSZ} for an example where the space is an extension field with its natural multiplication.
However, what is genuinely coding-theoretic is to consider minimum distance beside dimension.
It is well known that, asymptotically, a random code lies on the Gilbert-Varshamov bound $R=1-H(\delta)$.
It is then very natural to ask:
\begin{problem}
Does the product of two random codes, or the square or higher powers of a random code, lie on the GV bound?
\end{problem}
Observe that the answer would be negative if the question were stated with tensor product instead of $*$-product.

\end{document}